\documentclass[final,a4paper]{IEEEtran}
\usepackage{background}
\SetBgScale{0.8}
\SetBgAngle{0}
\SetBgOpacity{1}
\SetBgColor{blue}
\usepackage{cite,hyperref}
\SetBgPosition{current page.south}
\SetBgVshift{3cm}
\SetBgContents{\parbox{\textwidth}{\textcopyright{} 2019 IEEE.  Personal use of this material is permitted.  Permission from IEEE must be obtained for all other uses, in any current or future media, including reprinting/republishing this material for advertising or promotional purposes, creating new collective works, for resale or redistribution to servers or lists, or reuse of any copyrighted component of this work in other works.}}

\usepackage{algorithm, algpseudocode,graphicx}
\usepackage{amsmath,amsthm,amssymb,amsfonts,latexsym,mathrsfs}
\newcommand{\tmp}{Train Marshalling Problem}
\newcommand{\ground}{A(n)}

\theoremstyle{definition}
\newtheorem{theorem}{Theorem}

\newtheorem{lemma}[theorem]{Lemma}

\newtheorem{example}{Example}

\begin{document}
	\title{A Novel Dynamic Programming Approach
	to the Train Marshalling Problem}
	\author{Hossein~Falsafain
and Mohammad~Tamannaei
\thanks{H.~Falsafain is with the
	Department of Electrical and Computer Engineering, Isfahan University of Technology, Isfahan 84156-83111, Iran (e-mail: {h.falsafain@cc.iut.ac.ir}).}
\thanks{M.~Tamannaei
	is with the
	Department of Transportation Engineering, Isfahan University of Technology, Isfahan 84156-83111, Iran
	(e-mail: {m.tamannaei@cc.iut.ac.ir}).}}
\markboth{Accepted Article - IEEE TRANS. INTELL. TRANSP. SYST. - Falsafain \& Tamannaei - 
	\url{https://doi.org/10.1109/TITS.2019.2898476}}%
{Accepted Article - IEEE TRANS. INTELL. TRANSP. SYST. - Falsafain \& Tamannaei - 
\url{https://doi.org/10.1109/TITS.2019.2898476}}
\maketitle

	\begin{abstract}
Train marshalling
is the process of
reordering the railcars 
of a train
in such a way that
the railcars with the 
same destination appear
consecutively in 
the final, reassembled train.
The process takes place 
in the shunting yard
by
means of a number of classification tracks.
In the \tmp{} (TMP),
the objective 
is to perform this
rearrangement of the railcars
with the use of
 as few classification tracks as possible.
The problem
has been shown to be NP-hard, and
several exact and approximation algorithms
have been developed for it.
In this paper, we propose a
novel
exact dynamic programming (DP)
algorithm
for the TMP. The worst-case 
time
complexity of this algorithm
(which is exponential in the number of
destinations
and linear in the number of railcars)
 is lower than that of
  the best presently available
   algorithm for the problem,
   which is an inclusion-exclusion-based
   DP algorithm.
In practice,
the proposed algorithm 
can provide
a substantially improved performance
compared to its 
inclusion-exclusion-based counterpart,
as demonstrated
by the experimental results.
	\end{abstract}


	\section{Introduction} 
 In a shunting yard (a.k.a. classification or marshalling yard),
 the 
 railcars of an incoming (inbound)
 train are uncoupled,
 rearranged, and 
 then reassembled to 
 form an outgoing (outbound) train~\cite{orspec,ejor,haahr,dummies}.
 A shunting yard consists of multiple parallel 
 classification tracks
 (a.k.a. auxiliary rails) on which
 partial outbound trains can
 be assembled before being pulled together to form
 an outbound train.
 In a shunting yard,
one primary objective is to
rearrange the railcars of an inbound train
into groups that share the same destination 
\cite{haahr,rinaldi,dahlhaus,same-destination-1,same-destination-2}.
These groups, which we
refer to as \textit{blocks},
 are then coupled together to
 form a new outbound train.
 This process is referred to as
  \textit{train marshalling}, and 
 is accomplished by means of a number of 
 classification tracks.
While
a shunting yard 
has only a limited number
of classification tracks,
there may be several
 trains to be processed at
the same time.
Therefore, it is obvious that
the aim should be to keep 
the number of classification tracks 
per inbound train
as small as possible \cite{Adlbrecht,dahlhaus,min.tracks}.
In the \tmp{} (TMP),  
  using as few classification tracks as
  possible is the only objective \cite{ejor,beygang}.
 For a 
 detailed survey about other commonly used and
  more recent train classification methods
  (from an algorithmic point of view), see~\cite{dummies}.

Let an inbound train 
${\cal T}$
be given.
  In the TMP,
  in its optimization
  version,
  the goal is to 
find the minimum
number 
of classification tracks
needed to rearrange
the railcars of $\cal T$
  in such a way that the railcars
   sharing the same destination are grouped
   together.
     (Obviously, the 
   number of 
   classification tracks
   needed for a
   train 
   is never more than the number
   of destinations.)
   The decision version of
   the problem 
   can be stated as follows:
   Given an inbound train $\cal T$
    and a positive integer
   $k$,
   decide whether or not
   the railcars of $\cal T$ can properly be rearranged
   by means of at most $k$ classification tracks.

  The TMP has been
   proved to be NP-hard 
   (using a reduction from the
   numerical matching with target sums problem)
   \cite{dahlhaus}.
Generally, approaches for solving a particular
NP-hard problem can be classified into two main categories:
exact and inexact methods.
Exact methods are guaranteed to find an optimal solution if one exists
(see, e.g., \cite{rinaldi,tamannaei,shafia}).
However, they generally require exponential time in the worst case.
If the actual instances are small-to-moderate-sized, such algorithms may be
perfectly satisfactory.
The worst-case exponential complexity
is acceptable
if the algorithm is effective and fast enough
for the problem instances appearing 
in the specific application considered \cite{clrs,juraj}.
For some interesting  
discussions on exact algorithms in general, the reader is
referred to~\cite{exactexpo,hardsurvey,juraj}.
Inexact techniques, on the other hand, can find
reasonable suboptimal solutions in polynomial time.
Inexact methods themselves can be grouped into two families:
approximation algorithms, which guarantee to return a (suboptimal) solution
that is within a certain factor of the optimal solution
(see, e.g., \cite{approx1,approx2}),
and heurstic/metahueristic approaches, which do not offer any performance
guarantee but have been found
 to be very successful in 
 solving NP-hard problems (see, e.g., \cite{vns,ga,memetic,aco,haahr,tamannaei,shafia}).
It should be noted that some 
 tractable cases of the
 TMP have
 been identified in the literature \cite{sinica,dahl2}.
 In fact, some special cases of an NP-hard problem may 
 be solvable in polynomial time \cite{clrs}.
 (Reference~\cite{sinica} is
  the paper in which the TMP
 originally posed.)

 In~\cite{dahl2}, Dahlhaus et al. 
 showed that the problem 
 is  approximable within
 ratio $2$.
In \cite{beygang}, both online and offline versions of the TMP
have been considered.
In the offline scenario,
some basic results and lower bounds on the optimal solutions
have been presented.
Furthermore,
an analysis of 
the online version of the problem and  a 
2-competitive deterministic greedy online algorithm
have been provided.
It has also be shown that the competitive factor of 2 is indeed
best possible among all deterministic online algorithms.
 In \cite{fpt}, Brueggeman et al.
 established that the TMP is 
 fixed parameter tractable
 with respect to the 
 number of classification tracks $k$.
To be more precise,
if an inbound train 
$\cal T$ with $n$ railcars
 having $t$
  different destinations, and a positive integer $k$ are given,
then
for deciding whether or not 
the railcars of $\cal T$ can be rearranged in an
appropriate order by using at most $k$
classification tracks,
 the algorithm proposed
in \cite{fpt}
requires $O(2^{O(k)}\mathrm{poly}(n))$ time
and $O(n^2k2^{8k})$ space. 
The algorithm is 
based on the
dynamic programming
(DP)
 paradigm.
Finally, very recently, 
 Rinaldi and Rizzi 
 have developed 
 an exact dynamic programming algorithm for the problem
 in which
 the TMP
 is solved as a problem
 of finding a rainbow path
 in an edge-colored 
 directed graph~\cite{rinaldi}.
 This algorithm, which is based on 
 the principle of inclusion-exclusion,
 is of time complexity $O(nkt^22^t)$
 and
 space complexity
 $O(nkt)$.
 This implies that the TMP is
 fixed parameter tractable with the number of destinations $t$.
 The readily apparent advantage
 of this approach is its 
 polynomial space complexity.
 When used to solve the optimization
 version of the TMP, in a binary search fashion,
 the time complexity of 
 the approach becomes
 $O(nt^2 2^t U \log_2⁡U )$, where $U$ is an 
 upper bound on the optimal number of classification tracks.
 Although elegant, this procedure 
 only returns the
 value of an optimal solution,
 i.e. the minimum number of required classification tracks,
 but not an optimal solution itself.

In this paper, we propose a novel exact
DP algorithm
for the TMP. 
DP is indeed an approach of great importance in the design of both
polynomial-time and exponential-time algorithms \cite{clrs,exactexpo,juraj,hardsurvey}.
In contrast to the exact methods 
described in
\cite{rinaldi,fpt},
our method is developed to
\textit{directly} solve the
\textit{optimization version} of the TMP
(i.e.,
it does not need to
make
successive calls to
a procedure that solves
the decision version of
the problem).%
\footnote{It should be remembered that, there is not a
 significant difference between the optimization version of the
 TMP
 and its decision version.
 Obviously,
 an
 algorithm that can solve the optimization version
 of the problem,
 can automatically solve the decision
 version as well (for any given $k\in\mathbb{N}$).
 On the other hand,
 if one can solve the decision version 
 of the TMP for any given $k$,
 then one can also 
 minimize the number of classification tracks needed to obtain a
 train of desired property:
 For a given inbound train $\cal T$
with $n$ railcars having $t$ different destinations,
 the decision version of
 the problem is solved repeatedly,
 by incrementally 
 increasing the value of $k$
 from $1$. 
 When the answer 
 turns from ``no'' to ``yes'',
 then the solution is at hand.
 Furthermore,
 using binary search
 (instead of linear search),
 one needs to solve
  the decision version
  only
 for $O(\log_2 U)$ different values of $k$,
 where $U$ is an upper-bound on
 the optimal number of classification tracks.}
   The proposed algorithm is capable of finding not only 
   the minimum
   number of required classification tracks, but also
   the classification track that each railcar is
   assigned to, in an optimal solution, as well.
 The algorithm
 is of worst-case time complexity
 $O(nt 2^t)$
 and worst-case space complexity
  $O(n2^t )$.
  The Numerical
  experiments presented in Section~IV
  demonstrate that
  the algorithm
  substantially
  outperforms its
  inclusion-exclusion-based
  counterpart
  in terms of computation time.

The outline of this paper is 
as follows. 
In Section~II,
we 
provide some notations,
and
present a
rigorous
formulation of the TMP.
Section~III is dedicated
to the presentation
of our novel dynamic programming
algorithm for the TMP.
Section~IV
is devoted to experimental results,
and to
numerical comparisons with the
best currently available
approach to the problem.
(We extensively
compared our technique with
the approach presented in \cite{rinaldi}.)
 Finally, some concluding remarks are given
in Section~V.

\section{Notations and Problem Statement}

We begin this section by introducing some notation
that will be used throughout the paper.
For a positive integer $N\in {\mathbb{N}}$,
we denote by $[N]$ the set $\{1,2,\ldots,N\}$.
Each railcar of the train
is identified by its index.	
The order of 
the railcars in
the inbound train
corresponds therefore
to the sequence
$\langle1, 2, \ldots , n
\rangle$.
By a block,
we refer to 
an increasing sequence
consisting 
 of all
railcars 
that share the same
 destination.
A TMP instance
can
be described by a triple
$(n,t,{\cal B})$,
where $n$
is the number of
railcars of the train,
$t$ is the number of
destinations,
and
$\cal B$ is the set of all blocks.
(Obviously,
$|{\cal B}|=t$ and
$\bigcup_{B\in {\cal B}} B=[n]$.)
The objective is to find the smallest number of classification tracks
by which the inbound train can be rearranged according to the destinations.
 The railcars are considered one after another
according
to their order 
in the inbound train.
The railcars assigned to each classification track form a sequence.
In fact,
each railcar is guided to one of the 
classification tracks,
and placed
behind the  
already sequenced railcars.
The outbound train
is obtained by reassembling
the railcars on
the first track 
(based on their order of arrival)
followed by
the railcars on the second 
track and so on. 
\begin{example}
	\begin{figure}[htb]
		\centering
		\includegraphics[width=\linewidth]{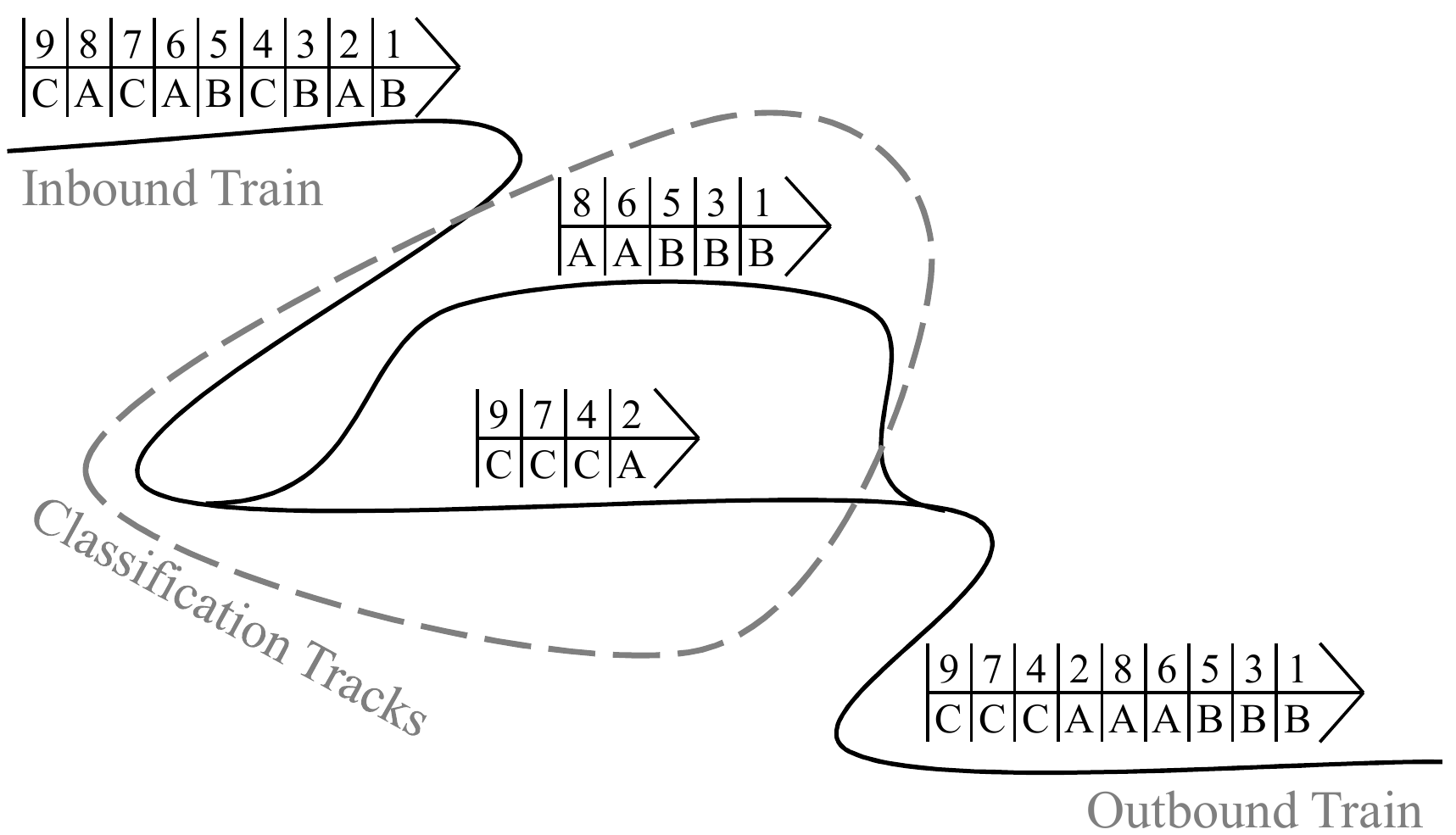}
		\caption{A depiction
			of the rearrangement
			process.
			The letter `A' 
			corresponds to the first destination,
			the letter `B' 
			corresponds to the second destination,
			and the letter `C' 
			corresponds to the third destination.}
		\label{illustrative}
	\end{figure}
	Figure~\ref{illustrative}
	depicts an instance of the problem in
	which
	$n=9$, $t=3$, and the set of
	blocks is
	$\mathcal{B}=\{\langle1,3,5\rangle,\langle2,6,8\rangle,
	\langle4,7,9\rangle\}$.
	As can be seen from the figure, to
	rearrange
	the railcars according to their destinations,
	only two classification tracks are needed.
\end{example}

The TMP
can be stated in an equivalent
way,
which is more convenient for
our purposes:%
\footnote{This
alternative statement of the problem
has been proposed by
Donald E. Knuth,
in a personal letter to the
authors of \cite{dahlhaus}.}
Let an instance
$(n,t,{\cal B})$
of the TMP
be given.
We define
$\ground$ to
be an infinite sequence 
obtained by concatenating infinitely many copies of the sequence
$\langle 1,2,\ldots,n\rangle$.
We refer to each subsequence
$\langle 1,2,\ldots,n\rangle$ of $\ground$
as a \textit{segment}.
(We will shortly see that
each segment corresponds to
a classification track.)
The entries of
$\ground$ can be seen as 
unoccupied positions
in which
the elements
of $\bigcup_{B\in{\cal B}}B=[n]$ can be \textit{placed}.
Each element of the set $[n]$
must occur exactly once in
$\ground$.
\textit{The occurrence of
the element $i\in[n]$
in the $\kappa$th segment
of $\ground$
is
equivalent to
the assignment of 
the $i$th railcar of the
inbound
 train
to the $\kappa$th classification track.}
It can readily be verified that
a placement 
in which
the elements of every
block
 $B\in {\cal B}$
occur  right next
 to each
other,
leads to
a solution to the given TMP instance,
and vice versa.
(We refer
to such a placement
as a solution as well.)
More precisely speaking,
a placement
can  be considered
as a solution
if and
only 
if for
every two
distinct blocks
$B$  and $B'$
 in $\cal B$,
none of the
elements of 
$B$ occur
in the interval occupied by the elements of
$B'$.
It is now clear that,
in the alternative statement of the TMP,
the aim is to find a 
placement
that
uses the minimum number of segments.
In other words,
the goal is to find 
a permutation
$\pi$	
of 
the set $[t]$
such that the
occurrence of
the blocks of $\cal B$
one after another,
in the order specified by
$\pi$,
uses the least number of 
segments.
In the following,
we stick to the above-described
alternative statement of the problem.

\begin{example}
Consider the following
instance of the TMP:
\begin{multline}\label{n17t5}
(n=17,\,t=5,\,
{\cal B}=\{B_1=\langle1, 4, 10\rangle, B_2=\langle5, 12, 14\rangle,\\
 B_3=\langle11, 13, 17\rangle, B_4=\langle2, 3, 8, 9\rangle, B_5=\langle6, 7, 15, 16\rangle\}).
\end{multline}
In a solution to
the above-given instance,
the elements of
$B_1$ and $B_2$ may appear in $A(17)$ in the following way:
\begin{multline*}\langle
\ldots,
17,
\underbrace{%
\mathbf{1},2,3,\mathbf{4},5,6,7,8,9,\mathbf{10}}_{B_1}
,11,\\
\underbrace{%
\mathbf{12},13,\mathbf{14},15,16,17,1,2,3,4,\mathbf{5}}_{B_2},
6,\ldots\rangle.
\end{multline*}
It should be remarked that
we defined a block as
an increasing sequence
consisting 
of all
railcars 
that share the same
destination,
but the element
of a block
$B\in \mathcal{B}$
do not necessarily appear in
$\ground$
in an increasing order.
The following 
is
an optimal solution
for the TMP instance described by
(\ref{n17t5}),
which uses three segments:
\begin{multline}\label{opt17}
\langle1,
\underbrace{
\mathbf{2},\mathbf{3},4,5,6,7,\mathbf{8},\mathbf{9}}_{B_4}
,10,
\underbrace{\mathbf{11},12,\mathbf{13},14,15,16,\mathbf{17}}_{B_3},\\
\underbrace{\mathbf{1},2,3,\mathbf{4},5,6,7,8,9,\mathbf{10}}_{B_1},
11,
\underbrace{\mathbf{12},13,\mathbf{14},15,16,17,1,2,3,4,\mathbf{5}}%
_{B_2},\\
\underbrace{\mathbf{6},\mathbf{7},8,9,10,11,12,13,14,\mathbf{15},
	\mathbf{16}}_{B_5}
,17,\ldots\rangle.
\end{multline}
This solution
is depicted schematically in Figure~2.
\end{example}

\begin{figure*}[htb]
	\centering
	\includegraphics[width=0.75\linewidth]{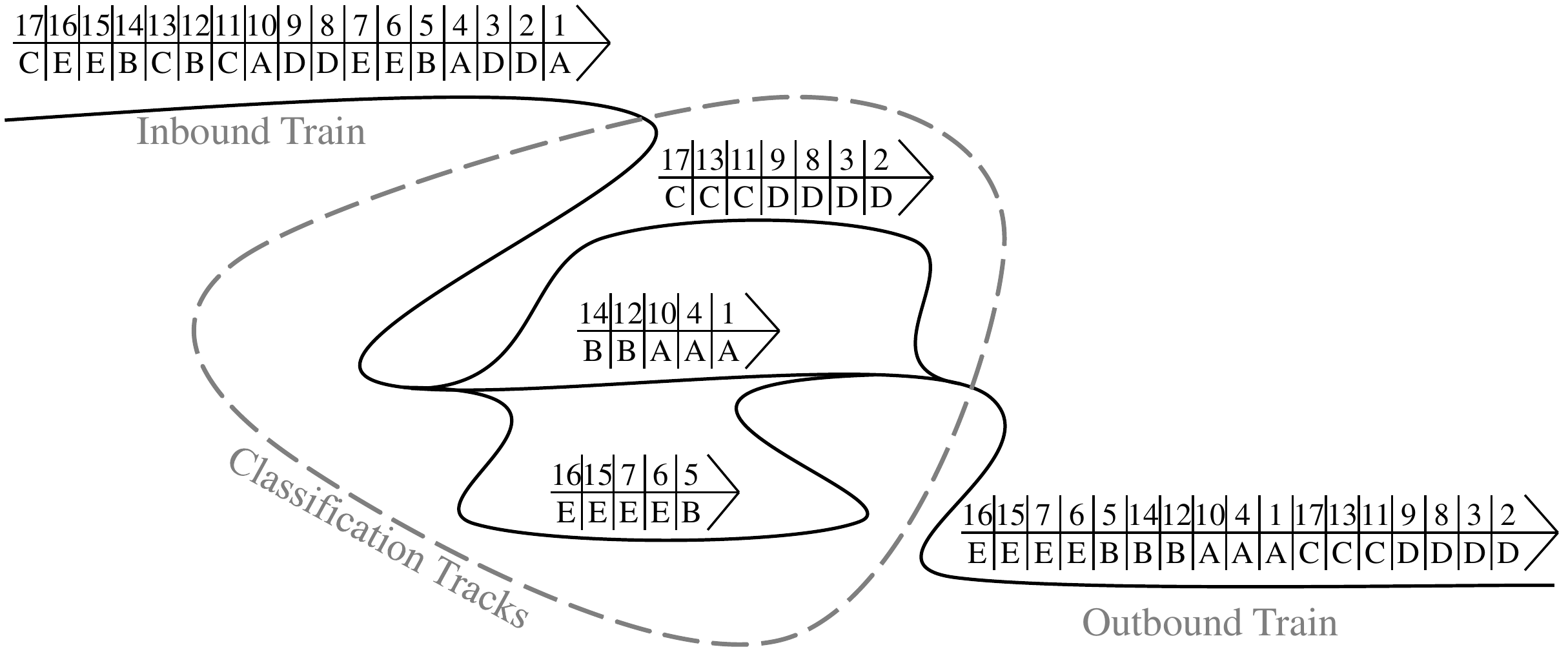}
	\caption{A depiction
	of the optimal solution
(\ref{opt17}) to the TMP
instance given in (\ref{n17t5}).
The letter `A' 
corresponds to the first block,
the letter `B' 
corresponds to the second block,
and so on.}
\end{figure*}

Before proceeding further,
we need to impose a number of 
simplifying assumptions,
which can be made without
loss of generality.
Let an instance $(n,t,{\cal B})$
of the TMP be given.
We can confine our attention to
those solutions
of this instance
in which, firstly,
there are no ``unused positions''	
before the first element of 
the first block of $\cal B$
that appears in $\ground$;
secondly,
there are no unused positions
between the elements of two
adjacent blocks;
and finally,
if $B$ is an arbitrary block in $\cal B$,
then the elements of
$B$ occur in $\ground$
contiguous with each other,
i.e., without any 
unused positions  between them.
Collectively,
the above three cases
indicate that we can, without loss
of generality,
narrow our attention to only
those solutions
in which there are no
positions left unused,
neither within the range of
appearance of the elements of a single
block
 $B\in {\cal B}$,
nor between 
the elements of two successive blocks,
nor before the first block
of $\cal B$
 that
appears in $\ground$.

\begin{example}
In a solution
to the TMP instance described by
(\ref{n17t5}),
the elements of
$B_4=\langle2, 3, 8, 9\rangle$
may appear in $A(17)$ as follows:
\begin{multline*}\langle
\ldots,17,1,
\mathbf{2},3,4,5,6,7,
\mathbf{8},
\mathbf{9},10,11,12,13,14,15,16,17,\\1,2,
\mathbf{3},4,\ldots\rangle.
\end{multline*}
But in this placement,
there is an unused position between
the elements $2$ and $8$,
because
the element $3$
has not been placed immediately
after $2$.
However,
placing the elements of $B_4$
one next to the other,
as required by the assumptions,
can lead to a 
solution
that is
at least as good:
\begin{multline*}\langle
\ldots,17,1,
\mathbf{2},\mathbf{3},4,5,6,7,
\mathbf{8},
\mathbf{9},10,11,12,13,14,15,16,17,\\1,2,
3,4,\ldots\rangle.
\end{multline*}

As another example,
in a solution
to the TMP instance
given by
(\ref{n17t5}),
the
elements of the blocks
$B_4$
and
$B_5$
can
appear in $A(17)$ as
follows:
\begin{multline*}\langle
\ldots,5,\underbrace{\mathbf{6},\mathbf{7},8,9,10,
	11,12,13,14,\mathbf{15},\mathbf{16}}_{B_5},17,\\
1,2,3,4,5,6,7,\underbrace{\mathbf{8},\mathbf{9},
	10,11,12,13,14,15,16,17,1,\mathbf{2},\mathbf{3}}_{B_4},4,\ldots\rangle.
\end{multline*}
But in this placement,
there are two unused positions
between 
$16$
(the last appearance of an element of $B_5$)
and 
$8$ (the first appearance of an element of $B_4$),
which can contain the elements $2$ and $3$ of
$B_4$.
By placing these two elements
immediately after
$16$,
we obtain a solution that 
satisfies our assumptions,
and is at least as good:
\begin{multline*}\langle
\ldots,5,\underbrace{\mathbf{6},\mathbf{7},8,9,10,11,12,13,
	14,\mathbf{15},\mathbf{16}}_{B_5},17,\\1,
\underbrace{\mathbf{2},\mathbf{3},4,5,6,7,
	\mathbf{8},\mathbf{9}}_{B_4},
10,\ldots\rangle.
\end{multline*}
\end{example}

We conclude 
this section by
introducing
two more notations.
Let an instance $(n,t,{\cal B})$
of the TMP be given, and
let $\underline{i}\in[n]$.
For reasons that will become clear
in what follows,
we 
may impose the condition that
none of the elements of $\bigcup_{B\in {\cal B}} B$
 could be placed
before the $\underline{i}$th
position of the first segment of $\ground$.
In the presence of  this condition,
we specify the instance by a quadruple
$(n,t,{\cal B},\underline{i})$.
We will
refer to such instances
as TMP instances as well.
(According to this notation,
a TMP instance $(n,t,{\cal B})$
can
be expressed as 
 $(n,t,{\cal B},1)$.)
Finally,
in this paper,
the notation `$+_n$'
symbolizes
 the usual  modulo-$n$ addition,
except that if the result is $0$,
it is replaced by $n$.
Therefore, for example,
$16+_{17}1=17$ and $17+_{17}1=1$.

	\section{A Dynamic Programming Approach to the \tmp{}}

	Let an
	instance
	$(n,t,{\cal B},\underline{i})$
	of the TMP be given,
	and let
	$\underline{B}=\langle 
	i_1,i_2,\ldots,i_{\ell}
	\rangle$ be the first 
	block of ${\cal B}$ that 
	appears in $\ground$.
	The first element of $\underline{B}$ 
	that appears in $\ground$ 
	is its smallest element that is  
	greater than or equal to
	$\underline{i}$.
	If no such element exists,
	then
	the first occurrence of an 
	element of
	$\underline{B}$
	in $\ground$ is  $i_1$,
	which in fact occurs in the second segment.
We
denote the first occurrence of an element of
$\underline{B}$ in $\ground$
 by $\alpha(\underline{i},\underline{B})$,
and the last occurrence of an element of
$\underline{B}$ in $\ground$ by
 $\omega(\underline{i},\underline{B})$.	
		Under the
 condition
 that
none of the elements of $\underline{B}$
can be placed
before the $\underline{i}$th
position of the first segment of $\ground$,
there are two 
different
cases to consider.
The aim in examining these
cases is to investigate what happens when
the first available
position in the first segment of
the sequence $\ground$
is $\underline{i}$, and
the block
whose
elements appear  immediately  following
$\underline{i}$
 is
$\underline{B}$.
The cases are as follows:
	\begin{itemize}
		\item Case~1, 
		$\underline{i}\leq i_1$
		or $i_{\ell}< \underline{i}$: 
		In this case,
		the first element of
		$\underline{B}$ that appears in $\ground$
		is $i_1$, the last element of
		$\underline{B}$ that appears in $\ground$
		is $i_{\ell}$, and 
		$i_{\ell}$ occurs in the same segment as $i_1$.
		According to the above-defined notations,
		we have
		$\alpha(\underline{i},\underline{B})=i_1$ and 
		$\omega(\underline{i},\underline{B})=i_{\ell}$.
		If~$i_{\ell}< \underline{i}$,
		then
		none of the elements of
		$\underline{B}$ can be placed in the first segment.
		Therefore, 
		$i_1$ and $i_{\ell}$ both
		occur in the second segment.
		\item Case~2,
		$i_1<\underline{i} \leq i_\ell$:
		In this case, we have 
		$i_l < \underline{i}\leq i_{l+1}$
		for some $l\in[\ell-1]$.
		The first element of
		$\underline{B}$ that appears in $\ground$
		is $i_{l+1}$ (i.e.,
		 $\alpha(\underline{i},\underline{B})=i_{l+1}$),
		 and
		the last element of
		$\underline{B}$ that appears in $\ground$
		is $i_{l}$
		(i.e., $\omega(\underline{i},\underline{B})=i_l$).
		Notice that
		$\alpha(\underline{i},\underline{B})=i_{l+1}$
		occurs in the first segment, but
		$\omega(\underline{i},\underline{B})=i_l$
		occurs in the second segment.
		Therefore, the elements of 
		$\underline{B}$ appear in $\ground$ in the following
		order:
		$$
		\underbrace{i_{l+1}, i_{l+2},\ldots,i_{\ell-1},
			i_{\ell}}_
		{\text{In the first segment}}
		,
		\underbrace{i_1,i_2,\ldots,i_{l}}_
		{\text{In the second segment}}.
		$$
	\end{itemize}

In Case~1,
if $i_{\ell}<n$, then
the block $\underline{B}$
does not
completely \textit{exhaust} 
its containing segment.	
This means that there
remain
$n-i_{\ell}$
positions available in the
segment (i.e., 
$i_{\ell}+1,i_{\ell}+2,\ldots,n$),
in which the elements
of the next
block may occur.
Otherwise (i.e., if 
$\underline{i}\leq i_1$ and
$i_\ell=n$),
the block $\underline{B}$
completely exhausts
its containing segment (which is the first segment).
This means that
the 
(yet-to-be-placed) block in
$\cal B$ that immediately follows
$\underline{B}$, starts somewhere
in the second segment.

To distinguish between the cases in which
$\underline{B}$
exhausts the first segment
and the cases in which
there remains at least one position
in the first
segment
that
can be utilized for the placement
of the block that follows $\underline{B}$,
we define an indicator
function
$\delta: [n]\times{\cal B} \mapsto \{0,1\}$
as follows.
(The symbol $\times$
denotes the 
Cartesian product.)
Let $\underline{i}\in[n]$	
and
$\underline{B}\in {\cal B}$
be given.
$\delta(\underline{i},\underline{B})$
takes the value~$1$
if $\omega(\underline{i},\underline{B})$ occurs
in the last position of the first segment
(i.e., $\omega(\underline{i},\underline{B})=n$)
or
occurs somewhere in the second segment,
and $0$ otherwise (i.e., if
$\omega(\underline{i},\underline{B})$
occurs in the first segment and
 $\omega(\underline{i},\underline{B})<n$).
Therefore, 
\begin{equation}\label{delta}
\delta(\underline{i},\underline{B})=
\begin{cases}
0, & \text{$\underline{i}\leq i_1$ and $i_\ell<n$},\\
1, & \text{otherwise}.\\
\end{cases}
\end{equation}

We are now ready to 
describe the algorithm.
Let $(n,t,{\cal B})$
be a given TMP instance.
Our algorithm
will make use of an
$n\times 2^t$ table $K$ whose rows are labeled by the elements of $[n]$, and
whose columns are labeled by the
subsets  of $\cal B$.
For a nonempty subset ${\cal B}'$
of ${\cal B}$,
and an integer $\underline{i}\in[n]$,
we define 
$K[\underline{i},{\cal B}']$
(the entry
at row $\underline{i}$ and column ${\cal B}'$
of $K$)
	to be the 
	minimum number of segments needed
	for placing the elements of 
	${\cal B}'$
	in $\ground$,
	with the restriction that
	the first
	element of $\bigcup_{B\in {\cal B}'} B$
	cannot appear before the $\underline{i}$th
	position of the first segment.
	It can clearly
	be seen from the above definition
	that
	the goal is 
	in fact
	to find $K[1,{\cal B}]$.
The most crucial component of a
 dynamic programming algorithm
is a recurrence relation that expresses 
the optimal
 solution 
 to an instance (recursively)
 in terms of optimal solutions to smaller 
 subinstances.
Our aim here is to
 derive such
a recurrence relation. 
In fact, although we are interested only in
$K[1,{\cal B}]$,	
we need to find the values of
the
entries of $K$
corresponding
to smaller instances of the problem. 
One question to be addressed here is
that
what does small mean in this context?
Some notion of the ``size'' of a subinstance
 is required here. 
 We sort the
 subinstances by size
 (from smallest to largest),
 and 
 solve them in
 increasing order of size.
 As we shall see shortly,
the only factor that determines the
size of the subinstance
corresponding to $K[\underline{i},{\cal B}']$
is indeed
the size of ${\cal B}'$.

	The objective here is to derive 
a recurrence for $K[\underline{i},{\cal B}']$.
To derive this  underlying
recurrence relation, we
need to consider all possible ways to choose 
the first block of ${\cal B}'$
that appears in $\ground$.
	If $\underline{B}$ is the first block of
	${\cal B}'$ that appears in $\ground$,
	then according to our definition of
	$K[\underline{i},{\cal B}']$,
	the 
	minimum number of segments needed
	for placing the elements of
	 the remaining blocks
	in the sequence $\ground$
	is 
	$K[\omega(\underline{i},\underline{B})+_n1,
	{\cal B}'\setminus \{\underline{B}\}]$.
	Therefore,
	under the restriction that
	none of the 
	elements of $\bigcup_{B\in {\cal B'}}B$
	is permitted to appear
	before the $\underline{i}$th
	position of the first segment,
	if $\underline{B}$ is the first block of
	${\cal B}'$ that appears in $\ground$,
	then the blocks of ${\cal B}'$
	together 
	occupy  
	$\delta(\underline{i},\underline{B})+
	K[\omega(\underline{i},\underline{B})+_n1,
	{\cal B}'\setminus \{\underline{B}\}]$
	segments.
Hence, we have the following recurrence: 
\begin{equation}\label{main.recurrence}
K[\underline{i},{\cal B}']=
\min_{\underline{B}\in {\cal B}'}\left\{
\delta(\underline{i},\underline{B})+
K[\omega(\underline{i},\underline{B})+_n1,{\cal B}'\setminus \{\underline{B}\}]
\right\}.
\end{equation}
Therefore,
the recursive property that yields the 
optimal 
solution value of the
original instance of the problem
is
$K[1,{\cal B}]=
\min_{\underline{B}\in {\cal B}}\left\{
\delta(1,\underline{B})+
	K[\omega(1,\underline{B})+_n1,
	{\cal B}\setminus \{\underline{B}\}]\right\}$.
This implies that
to obtain	
the value of an optimal solution
to the original instance
$(n,t,{\cal B},1)$,
we need to determine	
the optimal solution values of
the following $t$ subinstances:
$$(n-|\underline{B}|
,t-1,{\cal B}\setminus \{\underline{B}\},
\omega(1,\underline{B})+_n1),
\quad \text{
	for each
	$\underline{B}\in{\cal B}$}.$$
It remains  to specify the initial conditions.
We define the initial conditions as
$$
K[\underline{i},\varnothing]=
\begin{cases}
0,&\text{if $\underline{i}=1$},\\
1,&\text{otherwise}.
\end{cases}
$$
The rationale behind this definition
is as follows.
Assume 
that
in a solution
to a given TMP instance
$(n,t,{\cal B})$,
the
$(t-1)$th block
(i.e., the next-to-last block)
 of $\cal B$
that appears in $\ground$,
ends at the $i$th position
of the 
$\kappa$th segment, $i\in [n]$.
If
the last block of
$\cal B$ that
appears in $\ground$
is $\overline{B}\in {\cal B}$,
then we have 
the following equation,
which is obtained
by setting $\underline{i}=i+_n1$ and
${\cal B}'=\{\overline{B}\}$ in 
equation~(\ref{main.recurrence}):
\begin{equation}\label{init.cond}
K[i+_n1,\{\overline{B}\}]=
\delta(i+_n1,\overline{B})+
K[\omega(i+_n1,\overline{B})+_n1,\varnothing],
\end{equation}
There are three cases
to consider.
The first case indicates that
$K[1,\varnothing]=0$,
and the second and third cases together
imply that if $\underline{i}\neq 1$,
then $K[\underline{i},\varnothing]=1$:
\begin{itemize}
	\item
	If $\omega(i+_n1,\overline{B})=n$,
	then we necessarily have
	$K[i+_n1,\{\overline{B}\}]=1$.
	Moreover,
	according to
	Equation~(\ref{delta}),
	$\delta(i+_n1,\overline{B})=1$.
Now,
Equation~(\ref{init.cond}) implies that
	$K[\omega(i+_n1,\overline{B})+_n1,\varnothing]=0$.
	Since $\omega(i+_n1,\overline{B})=n$,
	we have $K[1,\varnothing]=0$.
	Notice that,
	in this case,
	the block $\overline{B}$
 starts  and ends at the $\kappa$th segment.
	\item
	If $\omega(i+_n1,\overline{B})< n$
	and
$\delta(i+_n1,\overline{B})=0$,
then
we necessarily have
$K[i+_n1,\{\overline{B}\}]=1$.
Now, Equation~(\ref{init.cond})
implies that
$K[\omega(i+_n1,\overline{B})+_n1,\varnothing]=1$.
In this case,
as in the first case,
the block $\overline{B}$
starts and ends at the $\kappa$th segment.

	\item
	If $\omega(i+_n1,\overline{B})< n$
and
$\delta(i+_n1,\overline{B})=1$,
then  we must have
$K[i+_n1,\{\overline{B}\}]=2$.
Again, Equation~(\ref{init.cond})
implies that
$K[\omega(i+_n1,\overline{B})+_n1,\varnothing]=1$.
In this case,
as opposed to the first two cases,
the block $\overline{B}$
either
starts at the $\kappa$th segment
and
ends at the $(\kappa+1)$th segment,
or
starts and ends at 
the $(\kappa+1)$th segment.
\end{itemize}

The (bottom-up table-based) version of our dynamic programming
algorithm for the Train Marshalling Problem,
the procedure 
\textsc{Bottom-Up-DP-TMP},
is presented in Algorithm~\ref{dp1}.
The correctness of this algorithm follows directly from (\ref{main.recurrence}).
Notice that in the table $K$, the only
entry in the column corresponding to 
$\cal B$ that we need to compute is 
$K[1, {\cal B}]$.
Hence, \textsc{Bottom-Up-DP-TMP} does
not compute all the entries
of this column,
but only $K[1, {\cal B}]$ (see line 16).
Although this procedure
determines the minimum number
of classification tracks needed
to rearrange the railcars in
an appropriate order, it
does not directly show how to 
rearrange them.
Therefore, besides the minimum number
 of needed classification
tracks, it returns 
a table $T$ using which an optimal
solution itself 
 can be constructed.
(This table 
provides us with
 the information we need to do
so.) $T$ has its rows
indexed by the
elements of $[n]$
and its columns indexed by the
subsets of $\cal B$.
The entry $T[\underline{i},{\cal B}']$
contains the first block that appears
 in $\ground$ in an optimal 
 solution to the TMP instance
$\left(\left|\bigcup_{B\in {\cal B}'}B\right|, 
\left|{\cal B}'\right|, {\cal B}', 
\underline{i}\right)$.
Using this table, the
procedure 
\textsc{Print-Optimal-Solution},
which is shown in
Algorithm~\ref{print},
prints out an optimal
 solution to a given TMP instance.
The rationale behind this procedure
is easy to grasp, so we omit a detailed discussion.

In \textsc{Bottom-Up-DP-TMP},
the time in both the first loop
(lines~5--7)
 and the last  loop  (line~16)
 is insignificant compared to the time in the
middle loop
(lines~8--15),
 because the middle loop contains various levels of nesting. (The loops are nested four deep.)
The time complexity of this nested loop
is 
$$
n\sum_{j=1}^{t-1}j\binom{t}{j}=
n\sum_{j=1}^{t-1}t\binom{t-1}{j-1}=
nt\sum_{j=0}^{t-2}\binom{t-1}{j}=nt(2^{t-1}-1).
$$
Therefore,
the time complexity of
the whole algorithm is also in $O(nt2^t)$.
It can easily be verified that
the
space complexity of 
the algorithm
is $O(n2^t)$.
If we
need only the \textit{value}
of an optimal solution,
and not an
optimal solution itself, 
then
we do not need to compute the 
table $T$ anymore.
Moreover,  in such a case,
if $K_j$
denotes the subtable
of $K$ consisting of columns
whose corresponding
subsets of $\cal B$
are of size $j$, $1\leq j \leq t$,
then
for computing
the entries of
$K_j$,
we only need 
the entries 
of $K_{j-1}$.
By virtue of this fact,
the memory
complexity
of
\textsc{Bottom-Up-DP-TMP}
can be reduced to
$O\left(n\binom{t}{\lfloor\frac{t}{2}\rfloor}\right)$.
(With respect to the facts that
the number of
$j$-subsets of $\cal B$ is $\binom{t}{j}$,
and that
$\binom{t}{\lfloor\frac{t}{2}\rfloor}$
is the largest of
the binomial coefficients $\binom{t}{j}$,
$0\leq j\leq t$.)

	\begin{algorithm*}
	\caption{A (bottom-up table-based) dynamic programming
		algorithm for the TMP.}\label{dp1}
	\begin{algorithmic}[1]
		\Statex \textbf{Input:}
		An instance $(n,t,{\cal B})$
		of the \tmp{};
		\Statex \textbf{Output:}
		The minimum
		number $k_{\mathrm{opt}}$
		of classification tracks 
		needed to obtain a
		desired
		rearrangement of the railcars, 
		and a (two-dimensional) table 
		$T$
		from which an optimal 
		solution
		can be constructed;
		\Procedure{Bottom-Up-DP-TMP}{$n,\,t,\,{\cal B}$}
		\State Let $K$ be a new table
		whose rows are labeled from
		$1$ through $n$,
		and whose columns are
		labeled by the subsets of $\cal B$;
		\State Let $T$ be a new table with
		the same dimensions as $K$;
		\State $K[1,\varnothing]=0$;
		\For{$\underline{i}=2\ \mbox{\bf to}\ n$}
		\State $K[\underline{i},\varnothing]=1$;
		\EndFor
		\For{$j=1\ \mbox{\bf to}\ t-1$}
		\For{all subsets ${\cal B}'$ of 
			$\cal B$ of size $j$}
		\For{$\underline{i}=1\ \mbox{\bf to}\ n$}
		
		\State $K[\underline{i},{\cal B}']=
		\min_{\underline{B}\in {\cal B}'}
		\left\{
		\delta(\underline{i},\underline{B})+K[
		\omega(\underline{i},\underline{B})+_n1,
		{\cal B}'\setminus \{\underline{B}\}]
		\right\}$;
		
		\State $T[\underline{i},{\cal B}']=$
		an arbitrary element of the set
		$\operatorname{argmin}_{\underline{B}
			\in {\cal B}'}\left\{
			\delta(\underline{i},\underline{B})+
			K[\omega(\underline{i},\underline{B})+_n1,
		{\cal B}'\setminus \{\underline{B}\}]
		\right\}$;
		\EndFor
		\EndFor
		\EndFor
		
		\State $k_{\mathrm{opt}}=
		\min_{\underline{B}\in {\cal B}}\left\{
			\delta(1,\underline{B})+
			K[\omega(1,\underline{B})+_n1,
		{\cal B}\setminus \{\underline{B}\}]
		\right\}$;
		\State $T[1,{\cal B}]=$
		an arbitrary element of the set
		$\operatorname{argmin}_{\underline{B}
			\in {\cal B}}\left\{
		\delta(1,\underline{B})+
		K[\omega(1,\underline{B})+_n1,
		{\cal B}\setminus \{\underline{B}\}]
		\right\}$;
		\State {\bf return} $k_{\mathrm{opt}}$ and $T$;
		\EndProcedure
	\end{algorithmic}

\end{algorithm*}

	\begin{algorithm*}
	\caption{A procedure for 
		constructing an optimal solution
		to a given TMP instance
		using the table $T$ returned by
		the procedure
		\textsc{Bottom-Up-DP-TMP}.
	}\label{print}
	\begin{algorithmic}[1]
		\Statex \textbf{Input:}
		An instance $(n,t,{\cal B})$
		of the \tmp{},
		and the table
		$T$ returned by the procedure
		\textsc{Bottom-Up-DP-TMP};
		\Statex \textbf{Output:}
		An optimal solution to the given TMP instance;
		\Procedure{Print-Optimal-Solution}{$n,\,t,\,{\cal B},\,T$}
		\State ${\cal B}'={\cal B};$
		\State $\underline{i}=1$;
		\While{${\cal B}'\neq \varnothing$}
		\State  $\underline{B}=T[\underline{i},{\cal B}']$;
		 \State Print $\underline{B}$;
		 \State ${\cal B}'={\cal B}'\setminus \{\underline{B}\}$;
		 \State $\underline{i}=\omega(\underline{i},\underline{B})+_n1$;
		 \EndWhile
		\EndProcedure
	\end{algorithmic}
	\end{algorithm*}
The time complexity 
of the
procedure
\textsc{Bottom-Up-DP-TMP}
is better than that of the
best currently known algorithm
for the TMP \cite{rinaldi}.
 However, it
 still has room for improvement.
We will now discuss some
modifications that  can be made
to improve the effectiveness of this
procedure.
We first state two lemmas.

\begin{lemma}\label{firstLemma}
	Let an instance $(n,t,{\cal B})$
	of the TMP be given.
	When
	applying 
	the procedure
	\textsc{Bottom-Up-DP-TMP}
	to
	this instance,
	for each ${\cal B}'\subseteq {\cal B}$,
	the 
	entries
	$$K[\underline{i}+_n1,{\cal B}'],
	\quad\underline{i}\in \bigcup_{B\in {\cal B}'}B,$$
	are not required 
	to be computed.	
\end{lemma}
\begin{proof}
	It can be observed from
	lines~11~and~16 of \textsc{Bottom-Up-DP-TMP}
	that
	for any $\underline{i}\in [n]$
	and any ${\cal B}'\subseteq {\cal B}$,
	the entry 
	$K[\underline{i},{\cal B}']$
	is required to be computed
	only if
	for some $h\in[n]$
	and some $B\in {\cal B}\setminus {\cal B}'$,
	we have $\underline{i}=\omega(h,B)+_n1$.
	It is not difficult
	to verify that 
	$
	\{\omega(h,B)+_n1\,|\,
	h\in[n],\,B\in {\cal B}\setminus {\cal B}'\}
	=\left\{h+_n1\,|\,h\in
	\bigcup_{B\in {\cal B}\setminus {\cal B}'}B  \right\}
	$.
	Now,
in the column corresponding
to ${\cal B}'$,
	none of 
	the entries 
	whose row indices
belong to the set
	$\left\{h+_n1\,|\,h\in
\bigcup_{B\in {\cal B}'}B  \right\}=
[n] \setminus
\left\{h+_n1\,|\,h\in
\bigcup_{B\in {\cal B}\setminus {\cal B}'}B  \right\}$
	need to be computed.
	\end{proof}

\begin{lemma}\label{secondLemma}
	Consider an instance $(n,t,{\cal B})$
	of the TMP.
	Let ${\cal B}'\subseteq {\cal B}$ be given,
	and
	let $\ell=\sum_{B\in{\cal B}'}|B|$.
	($|B|$ denotes the length of
	$B$.)
	Let
	$\sigma({\cal B}')=\left\langle i_1,i_2,\ldots,
	i_{\ell}\right\rangle$
	 be an increasing sequence whose
	elements are the elements of
	$\bigcup_{B\in {\cal B}'}B$.
	If $l\in[\ell-1]$,
	then for all distinct integers
	$i_l < \underline{i},\underline{i'} \leq i_{l+1}$
	we have
	$K[\underline{i},{\cal B}']=K[\underline{i'},{\cal B}']$.
	This fact holds also for
	all distinct integers
	$1\leq \underline{i},\underline{i'} \leq i_1$,
	and
all distinct integers	
	$i_\ell< \underline{i},\underline{i'}\leq n$.
\end{lemma}
\begin{proof}
Let
$i_l < \underline{i}\neq\underline{i'} \leq i_{l+1}$,
for some $l\in[\ell-1]$.
In the optimal solutions
corresponding to
the entries
$K[\underline{i},{\cal B}']$
and
$K[\underline{i'},{\cal B}']$,
no matter 
which block occurs 
first,
the first element of $\bigcup_{B\in {\cal B}'}B$
cannot appear before the 
$i_{l+1}$th
position
of
the first segment.
Therefore,
due to
our definition of
the table $K$,
it can be verified that
$K[\underline{i},{\cal B}']
=K[\underline{i'},{\cal B}']
=K[i_{l+1},{\cal B}']$.
The proof of the remaining assertions of the
lemma 
(i.e.,
$K[\underline{i},{\cal B}']=K[\underline{i'},{\cal B}']$
for $1\leq \underline{i}\neq \underline{i'} \leq i_1$, and 
$K[\underline{i},{\cal B}']=K[\underline{i'},{\cal B}']$
for $i_\ell< \underline{i}\neq \underline{i'}\leq n$)
is analogous.
\end{proof}

The first lemma asserts
that,
in the column 
of the 
table $K$
corresponding 
to ${\cal B}'$,
the entries of the rows
whose indices
belong to
$\left\{h+_n1\,|\,
h\in \bigcup_{B\in {\cal B}'}B\right\}$
do not need to
be computed.
The second lemma states
that, 
if
$\sigma({\cal B}')=\left\langle i_1,i_2,\ldots,
i_{\ell}\right\rangle$
is an increasing sequence whose
elements are the elements of
$\bigcup_{B\in {\cal B}'}B$, then
in the column 
of the 
table $K$
corresponding 
to ${\cal B}'$,
the entries of
the rows whose indices
belong to
$\{1,2,\ldots,i_1\}$
are all equal to each other,
the entries of
the rows whose indices
belong to	
$\{i_l+2,i_l+3\ldots,i_{l+1}\}$,
$l\in[\ell-1]$,
are all equal to each other,
and
the entries of
the rows whose indices
belong to
$\{i_{\ell}+2,i_{\ell}+3\ldots,n\}$
are all equal to each other.
It should be highlighted that,
if 
for some $l\in[\ell-1]$,
$i_l$ and $i_{l+1}$
are two consecutive integers,
then
$\{i_l+2,i_l+3\ldots,i_{l+1}\}$
is empty.
Therefore,
we do not need to do anything about  it,
which means
even
 less computational effort.
Similarly, if
$i_\ell=n-1$, then
$\{i_{\ell}+2,i_{\ell}+3\ldots,n\}$
is empty, and the same conclusion holds.
\begin{example}
	Consider the TMP instance
	given by
	(\ref{n17t5}).
	For the
	subset
	${\cal B}'=\{B_1,B_2,B_3\}$
	of ${\cal B}$,
	we have
	$
	\sigma({\cal B}')=\langle1,4,5,10,11,12,13,14,17\rangle.
	$
Therefore, only four 
entries in the column corresponding to
${\cal B}'$ need be computed,
one corresponding to
each of the sets
$\{1\}$, $\{3,4\}$, $\{7,8,9,10\}$,
and $\{16,17\}$.
The other entries in 
this column
either
do not need to be computed at all or
require to simply be 
retrieved according to Lemma~\ref{secondLemma}.
Without the aid of 
the above two lemmas,
we need to obtain all
$n=17$ entries in this column.
\end{example}

Now we are ready
to present the improved 
version of our algorithm.
The improved version, the
procedure 
\textsc{Memoized-DP-TMP},
which is
detailed in Algorithm~\ref{memoized},
is indeed a  top-down (recursive)
\textit{memoized} version
of the procedure
\textsc{Bottom-Up-DP-TMP}
that
takes advantage of 
the results stated in
Lemmas~\ref{firstLemma} and \ref{secondLemma}.
In a bottom-up dynamic-programming algorithm,
we fill a table with 
solutions to \textit{all} smaller subinstances.
But 
solutions to some of
these smaller subinstances
may not be necessarily required
for obtaining a solution
to the original instance.
In fact,
in a DP framework,
we do not necessarily
need to solve all the subinstances 
in order to find an
optimal solution to the original problem.
It is natural to try  to develop
 an
 improved
  mechanism  that solves only
  those subinstances
  that
are necessarily needed. 
This approach is called \textit{Memoization}%
~\cite{clrs}.
The procedure \textsc{Memoized-DP-TMP},
shown in Algorithm~\ref{memoized},
exactly as in the procedure 
\textsc{Bottom-Up-DP-TMP},
makes use of a table,
but it computes 
the entries of
this table
in an as-needed fashion.
The table
 is initially filled with $-1$s.
 The value $-1$
  indicates that 
  the corresponding subinstance has
  not yet been solved.
  Indeed,
  the table entries initially 
  contain $-1$ values to indicate that
they have not yet been filled in.
  Whenever the optimal value
  of a subinstance
   needs to be obtained,
    the procedure \textsc{Lookup},
    shown in Algorithm~\ref{memoized},
       checks the corresponding 
    entry in the table $K$ first. 
    If the entry is not equal to 
    $-1$,
    it is simply 
    retrieved from the
  table
  (lines~2--4
  of the procedure
  \textsc{Lookup});
  otherwise, it is computed 
  by making recursive calls 
  (line~5
  of the procedure
  \textsc{Lookup}).
  The
   result is then recorded
  in the table 
  with respect to 
  Lemmas~\ref{firstLemma}~and~\ref{secondLemma}.
   This means that
   when the value of the
   $(\underline{i},{\cal B}')$-entry of $K$
   is computed recursively,
   then the result
    is stored not only 
   in the entry itself, but also
   in the entries whose
   values are equal to that
   of the $(\underline{i},{\cal B}')$-entry
   by virtue of 
   Lemmas~\ref{firstLemma}~and~\ref{secondLemma}
   (lines~8--20
   of the procedure
   \textsc{Lookup}). 
   Notice that
   the procedure
\textsc{Memoized-DP-TMP}
only returns the value of
an optimal solution.   
This is for the sake of brevity and clarity.
It can easily be 
modified so that it can 
produce 
the table $T$,
which is
required 
as an input to 
\textsc{Print-Optimal-Solution}, as well.

	\begin{algorithm*}
	\caption{A memoized top-down dynamic programming
		algorithm for the TMP.}\label{memoized}
	\begin{algorithmic}[1]
		\Statex \textbf{Input:}
		An instance $(n,t,{\cal B})$
		of the \tmp{};
		\Statex \textbf{Output:}
		The minimum
		number 
		of classification tracks 
		needed to obtain a
		desired
		rearrangement of the railcars;
		\Procedure{Memoized-DP-TMP}{$n,t,{\cal B}$}
		\State Let $K$ be a new table
		whose rows are labeled from
		$1$ through $n$,
		and whose columns are
		labeled by the subsets of $\cal B$;
		\State Initialize all entries of $K$ to $-1$;
		\State $K[1,\varnothing]=0$;
		\For{$i=2\ \mbox{\bf to}\ n$}
		\State $K[i,\varnothing]=1$;
		\EndFor
		\State
		{\bf return} $\min_{
			\underline{B}\in {\cal B}}
		\left\{
		\delta(1,\underline{B})+
		\text{\textsc{Lookup}}(K,\omega(1,\underline{B})+_n1,
		{\cal B}\setminus \{\underline{B}\})
		\right\}$;
		\EndProcedure	
	\end{algorithmic}
	
	\hrulefill
	
	\begin{algorithmic}[1]
		\Procedure{Lookup}{$K,\underline{i},{\cal B}'$}
		\If{$K[\underline{i},{\cal B}']\neq -1$}
		\State {\bf return} $K[\underline{i},{\cal B}'];$
		\EndIf
		\State
		$k=
		\min_{\underline{B}\in {\cal B}'}
		\left\{
		\delta(\underline{i},\underline{B})
		+\text{\textsc{Lookup}}
		(K,\omega(\underline{i},\underline{B})+_n1,
		{\cal B}'\setminus \{\underline{B}\})
		\right\}$;
		
		\State Let $\ell=\sum_{B\in{\cal B}'}|B|$;
		\State Let
		$\sigma=\left\langle i_1,i_2,\ldots,
		i_{\ell}\right\rangle$
		be an increasing sequence whose
		elements are the elements of
		$\bigcup_{B\in {\cal B}'}B$;
		\If{$1\leq \underline{i} \leq i_1$}
		\For{$i=1\ \text{\bf to}\ i_1$}
		\State $K[i,{\cal B}']=k$;
		\EndFor
		\ElsIf{$i_l+2 \leq \underline{i} \leq i_{l+1}$
			for some 
			$1\leq l \leq \ell-1$}
		
		\For{$i=i_l+2\ \text{\bf to}\ i_{l+1}$}
		\State $K[i,{\cal B}']=k$;
		\EndFor
		
		\Else
		\For{$i=i_\ell+2\ \text{\bf to}\ n$}
		\State $K[i,{\cal B}']=k$;
		\EndFor
		
		\EndIf
		
		\State {\bf return} $k$;
		\EndProcedure	
	\end{algorithmic}
	
\end{algorithm*}

\section{Computational Results}
	
 This section is devoted to evaluating the 
 performance 
  of the proposed technique
against  the best currently available
approach
to the TMP,
 which is a DP
procedure based on the 
principle of inclusion-exclusion,
proposed in~\cite{rinaldi}.
The
comparison has been made using
540 randomly generated instances
of the TMP, in its optimization version
(10 instances
for each considered value of
$n$ and $t$).%
\footnote{%
The literature on the
\tmp{}
 is relatively recent. 
 To our knowledge,
 there is no benchmark dataset available
 in the literature for the problem.
 We therefore created our own dataset,
 which introduces a rather wide range
 of TMP instances, with varying values of
 $n$ and $t$.
	We have made all
	the problem 
	instances,
	as well as their optimal solutions, 
	available
	online on
	the webpage
	 \url{https://github.com/hfalsafain/Train-Marshalling-Problem}.
	 All codes
will also be made available to download
after publication.
}
No experimental results have been reported 
in
\cite{rinaldi}.
The following results are therefore
based on our own implementation of
the algorithm described in \cite{rinaldi}.
Care was taken to implement the algorithms as efficiently as possible.
However,
this certainly does not mean that there is no room for improvement.
The algorithms
have been implemented in
Maple~18.00,
and the experiments have been
carried out on an Intel~Core~i5-3330~CPU
at~3.0--3.2~GHz desktop
computer with
4.00~GB of RAM, running Microsoft~Windows~8.1 operating system.

Before going further,
let us describe a simple 
preprocessing step that,
without altering
the optimal value,
can reduce the 
number of railcars in
 the 
input instance (see \cite[Lemma~1]{rinaldi}). 
This consequently can reduce
the execution time of the algorithms.
	Let an instance
	$(n,t,{\cal B})$ of
	the TMP be given. If there
	exist 
	$n'\geq2$ consecutive railcars
	that all share the same destination,
	then the instance
	$(n-n'+1,t,{\cal B}')$
	obtained by keeping
	only
	one of these railcars
	and removing the others,
	has the same optimal value as the original.
	Therefore, we can
	safely shrink the 
	number of railcars in
	the given instance by
	repeatedly eliminating
	all such railcars.

\begin{table*}[thb]
	\caption{A comparison of the average
	running times 
	(in seconds)
	of
	our method
	(the procedure \textsc{Memoized-DP-TMP})
	against those of the 
	method of~\cite{rinaldi}.
	The acronym
	``TLE''
	stands for ``Time Limit Exceeded''}
\label{main.table}	
	\centering
	
	\footnotesize
	\setlength{\tabcolsep}{3.5pt}
	\begin{tabular}{|c|c|cccccc|}
		\hline
		$n$ & Approach& $t=5$ & $t=7$&$t=9$&$t=11$& $t=13$ & $t=15$\\
		\hline
		50 & Rinaldi \& Rizzi (2017)&
		12.019& 153.997 & 1114.678 &TLE&TLE&TLE\\
		&Our Algorithm~3 & 0.042 &
		0.196 & 1.228 & 7.053 & 28.991 & 144.902\\
		\hline
		75& Rinaldi \& Rizzi (2017) &
		18.067 & 249.447 & 1964.511&TLE&TLE&TLE\\
		&Our Algorithm~3 & 0.053 & 0.245&
		1.801 & 11.209& 49.844 & 269.347\\
		\hline
		100& Rinaldi \& Rizzi (2017)&24.219 & 337.292 & 3288.679&TLE&TLE&TLE\\
		&Our Algorithm~3&0.061&0.294&2.225&15.588&72.216 & 470.611\\
		\hline
		200& Rinaldi \& Rizzi (2017) & 50.036 & 710.131&TLE&TLE&TLE&TLE\\
		&Our Algorithm~3&0.075&0.522&3.775&27.152&167.992&1157.106\\
		\hline
		300& Rinaldi \& Rizzi (2017) & 78.106 &1131.708&TLE&TLE&TLE&TLE\\
		&Our Algorithm~3 & 0.103
		&0.678&5.260&30.840&227.534&1595.130\\
		\hline
		400& Rinaldi \& Rizzi (2017) & 108.283&1599.885&TLE&TLE&TLE&TLE\\
		&Our Algorithm~3&
		0.119&0.861&7.044&41.290&281.694
		&1918.972\\
		\hline
		500& Rinaldi \& Rizzi (2017) & 140.439&2112.443&TLE&TLE&TLE&TLE\\
		&Our Algorithm~3&0.139&1.055&8.602&48.536&348.833&2341.133\\
		\hline
		750& Rinaldi \& Rizzi (2017) & 231.226&3248.614&TLE&TLE&TLE&TLE\\
		&Our Algorithm~3&0.206&1.694&12.958&72.902&493.286&3438.906\\
		\hline
		1000& Rinaldi \& Rizzi (2017) & 334.431&4825.828&TLE&TLE&TLE&TLE\\
		&Our Algorithm~3&0.288&2.336&17.233&97.056&677.042&4565.083\\
		\hline
	\end{tabular}

\end{table*}

A comparison of the average
running times 
(in seconds)
of
our method
(the procedure \textsc{Memoized-DP-TMP})
against those of the 
method of~\cite{rinaldi}
is tabulated in
Table~\ref{main.table}.
As stated above,
for each considered value of $n$ and $t$,
10 instances have been solved.
The table reports
the average solution time per instance.
In all runs, we 
imposed
a time limit of 5000 seconds.
In the table, ``TLE''
stands for ``Time Limit Exceeded.''
It should be noticed that
the method of \cite{rinaldi}
is inherently
designed for 
solving the decision version of
the TMP.
But 
it can be employed for solving
the optimization version of the
problem as well, by means of
a binary search procedure,
as stated in Section~I.
The resulting algorithm
is of time complexity $O(nt^22^tU\log_2U)$,
where $U$ is an upper-bound on
the optimal number of classification tracks.
In our implementation of the approach of \cite{rinaldi},
we have used the upper bound
$U=\min\left\{t,\left\lceil
\frac{n}{4}+\frac{1}{2}
\right\rceil\right\}$.
(Trivially, the
minimum number of classification tracks needed for
rearranging the railcars in an appropriate order
is at most $t$.
Furthermore,
it has been shown in \cite[Section~3]{dahlhaus}
that
the value
of an optimal solution to
a TMP instance
with $n$ railcars is at most 
$\left\lceil
\frac{n}{4}+\frac{1}{2}
\right\rceil$. See also \cite[Theorem~3]{rinaldi}.)	
Using a more
restrictive upper bound,
can obviously reduce
the number of calls
to a procedure for 
solving the decision problem
(during the binary search procedure).
However,
our Algorithm~\ref{memoized}
requires much less
computation
 time, even compared to the time required
for a single call to
the algorithm given in~\cite{rinaldi} for
solving a specific decision
 problem instance 
 (i.e., for a specific value of $k$).
For some of
the considered values of $n$ and $t$,
Table~\ref{table.second}
provides a comparison
of the average running times of
our procedure \textsc{Memoized-DP-TMP}
with those of
each of the calls
to the procedure described in~\cite{rinaldi}
for solving
the decision version of the problem
(called as \textsc{DTMP}
 in \cite[Algorithm~3]{rinaldi}).
As is evident from 
both Tables~\ref{main.table} and \ref{table.second}, our algorithm
substantially
outperforms its
inclusion-exclusion-based
counterpart
in terms of computation time.
	
We conclude this section
by making a remark
concerning the space complexity of
our approach.
Given a TMP instance $(n, t,{\cal B})$ and 
given
$k\in {\mathbb{N}}$,
for deciding whether or not 
the railcars can be rearranged in an
appropriate order by using at most $k$
classification tracks,
the algorithm described in \cite{rinaldi}
requires $O(nkt^22^t)$ time
and $O(nkt)$ space.
On the other hand,
our approach,
in its worst-case behavior,
requires 
$O(nt2^t)$ time
and $O(n2^t)$ space
for making such a decision.
Therefore,
from the time complexity
point of view,
our approach is superior
to the approach described in \cite{rinaldi}.
However,
in contrast to the method of \cite{rinaldi},
our algorithm requires exponential space
with respect to $t$.%
\footnote{%
As stated in Section~I,
the method of \cite{rinaldi}
	does not return an
	optimal
	solution itself.
	In fact,
	this algorithm,
	when used for solving the 
	decision version
	of the TMP,
	only returns
	a yes/no answer,
	and when used for solving the 
	optimization version
	of the TMP (in the
	manner descrined in Section~I),
	only computes
	the \textit{value} of an optimal solution.
	On the other hand,
	as has been discussed in the previous section,
	our approach
	not only returns the value of
	an optimal solution,
	but also
	the solution itself.
	Although
	typically a dynamic programming
	algorithm that only
	returns the 
	value of an optimal solution
	can easily be modified 
	in such a way that it can
	return an optimal solution
	itself~\cite[Chapter~15]{clrs},
it seems to us that
modification of the method of
\cite{rinaldi}
so that it can 
construct an optimal solution itself
leads to an exponential 
growth in 
the space complexity.}

\begin{table}[thb]
	\caption{A comparison
		of the average running times of
		our procedure \textsc{Memoized-DP-TMP}
		with those of
		each of the calls
		to the procedure \textsc{DTMP} 
		presented in
		\cite[Algorithm~3]{rinaldi}}
	\label{table.second}
\centering
\footnotesize
\begin{tabular}{|cc|c|c|c|}
	\hline
	$n$ & $t$ & \multicolumn{2}{c|}{Approach}
	& Time (in secs.) \\
	\hline
	50 & 5  && $k=2$&2.161\\
	& &DTMP~\cite[Algorithm~3]{rinaldi} &$k=3$   &3.994\\
	& & &$k=4$   &5.864\\
	\cline{3-4}
	& & 
	\multicolumn{2}{c|}{\textsc{Memoized-DP-TMP}} & \textbf{0.042} \\
	\hline
	500 & 5  && $k=2$&22.655\\
& &DTMP~\cite[Algorithm~3]{rinaldi} &$k=3$   &46.047\\
& & &$k=4$   &71.738\\
\cline{3-4}
& & 
\multicolumn{2}{c|}{\textsc{Memoized-DP-TMP}} & \textbf{0.139} \\
\hline
	1000 & 5  && $k=2$&51.017\\
& &DTMP~\cite[Algorithm~3]{rinaldi} &$k=3$   &107.788\\
& & &$k=4$   &175.627\\
\cline{3-4}
& & 
\multicolumn{2}{c|}{\textsc{Memoized-DP-TMP}} & \textbf{0.288} \\
\hline
50 & 7  && $k=3$&30.855\\
& &DTMP~\cite[Algorithm~3]{rinaldi} &$k=4$   &45.442\\
& & &$k=5$   &60.072\\
& & &$k=6$   &74.823\\
\cline{3-4}
& & 
\multicolumn{2}{c|}{\textsc{Memoized-DP-TMP}} & \textbf{0.196} \\
\hline
500 & 7  && $k=3$&345.878\\
& &DTMP~\cite[Algorithm~3]{rinaldi} &$k=5$   &765.925\\
& & &$k=6$   &1000.639\\
\cline{3-4}
& & 
\multicolumn{2}{c|}{\textsc{Memoized-DP-TMP}} & \textbf{1.055} \\
\hline
1000 & 7  && $k=3$&737.058\\
& &DTMP~\cite[Algorithm~3]{rinaldi} &$k=5$   &1738.398\\
& & &$k=6$   &2350.372\\
\cline{3-4}
& & 
\multicolumn{2}{c|}{\textsc{Memoized-DP-TMP}} & \textbf{2.336} \\
\hline
50 & 9  && $k=4$&269.248\\
& &DTMP~\cite[Algorithm~3]{rinaldi} &$k=5$   &359.379\\
& & &$k=6$   &446.749\\
& & &$k=7$   &555.890\\
\cline{3-4}
& & 
\multicolumn{2}{c|}{\textsc{Memoized-DP-TMP}} & \textbf{1.228} \\
\hline
100 & 9  && $k=4$&526.833\\
& &DTMP~\cite[Algorithm~3]{rinaldi} &$k=6$   &897.484\\
& & &$k=7$   &1090.555\\
& & &$k=8$   &1289.678\\
\cline{3-4}
& & 
\multicolumn{2}{c|}{\textsc{Memoized-DP-TMP}} & \textbf{2.225} \\
\hline
\end{tabular}

\end{table}

	\section{Conclusions}
	In this contribution, we developed a
	novel dynamic programming approach
	to the \tmp{} (TMP)
	whose worst-case time complexity is
	linear in the number of railcars,
	exponential in the number of destinations.
	One noticeable 
	difference between
	our approach
	and the previously proposed
	exact approaches to the TMP
	is that,
	in contrast with
	the previous works,
	our algorithm
	is designed to deal directly with the
	optimization version of the problem.
The worst-case time complexity of 
our method, like that of previous exact methods,
is exponential with respect to
the number of destinations.
However,
 in practice,
 our approach performs substantially
 better than the best currently available approach,
 which is an inclusion-exclusion-based 
 dynamic programming algorithm.
Our algorithm 
	can effectively tackle instances with
	relatively 
	large values of $t$.
	For example,
	it can solve, to optimality,
	instances involving
	$1000$ railcars and $15$
	destinations in about $1\frac{1}{4}$ hours.
The superior performance 
of the proposed technique
can mainly be
attributed to the following two reasons.
Firstly, we group together
the subinstances that
have the same optimal solution.
The optimal solution
is computed only once for each group.
Secondly, we employ the memoization
technique
to solve only
those subinstances whose optimal solutions
are necessarily needed to solve
the original instance.
In fact,
the main benefit of this technique is 
that only those table entries
that are needed are computed,
whereas in the 
bottom-up implementation, all
table entries get computed blindly.	
This can
significantly reduce the amount of
computation necessary.
Another advantage of our method compared to its
inclusion-exclusion-based counterpart
is that
our exact algorithm is capable of finding not only the
value of an optimal solution
(i.e., the minimum  number of required classification tracks),
but also the solution itself as well.

\bibliographystyle{ieeetr}

\end{document}